\documentclass[12pt, oneside]{article}   	
\usepackage[margin=2cm, includeheadfoot]{geometry}
\usepackage{authblk}
\usepackage[round, semicolon]{natbib}

\usepackage[usenames,dvipsnames]{xcolor}
\usepackage{hyperref}
\definecolor{plum}  {rgb}{.4,0,.4}
\definecolor{bred} {rgb}{0.6,0,0}
\hypersetup{colorlinks = true, linkcolor = bred, citecolor = blue, urlcolor=urlcolr}
\usepackage{enumitem}
\usepackage[utf8]{inputenc} 
\usepackage[T1]{fontenc} 
\usepackage{url}
\usepackage{ifthen}
\usepackage{cite}
\usepackage[cmex10]{amsmath} 

\usepackage{amsfonts}
\usepackage{amsthm}    
\usepackage{url}
\usepackage{booktabs}

\DeclareMathOperator{\MMI}{MMI}
\DeclareMathOperator{\BROJA}{BROJA}
\DeclareMathOperator{\red}{red}
\DeclareMathOperator{\ccs}{ccs}
\DeclareMathOperator{\dep}{dep}
\DeclareMathOperator{\IG}{IG}
\DeclareMathOperator{\GH}{GH}

\DeclareMathOperator{\conv}{conv}
\DeclareMathOperator*{\argmin}{arg\,min}
\DeclareMathOperator*{\argmax}{arg\,max}

\newcommand{\Scal}{\mathcal{S}}

\newcommand{\Ycal}{\mathcal{Y}}
\newcommand{\Zcal}{\mathcal{Z}}
\newcommand{\Pb}{\mathbb{P}}
\newcommand{\Rb}{\mathbb{R}}

\newtheorem{lemma}{Lemma} 
\newtheorem{theorem}[lemma]{Theorem}

\newtheorem{definition}[lemma]{Definition}
\newtheorem{example}[lemma]{Example}

\definecolor{gmcolor}{rgb}{0.6, 0.2, 0.9}

\title{\bf Continuity and Additivity Properties\\ of Information Decompositions}

\author[1]{Johannes Rauh}
\author[2]{Pradeep Kr.~Banerjee}
\author[2]{Eckehard Olbrich}
\author[2,3]{Guido Montúfar}
\author[2]{Jürgen Jost}

\affil[1]{Institut f{\"u}r Qualit{\"a}t und Transparenz im Gesundheitswesen (IQTiG),  Berlin}
\affil[ ]{\textit {jarauh@gmx.net}}
\affil[2]{Max-Planck-Institut f{\"u}r Mathematik in den Naturwissenschaften (MiS), Leipzig}
\affil[ ]{\textit {\{pradeep,olbrich,jjost\}@mis.mpg.de}}
\affil[3]{University of California, Los Angeles}
\affil[ ]{\textit {montufar@math.ucla.edu}}

\date{}

\begin{document}
	\maketitle

\begin{abstract}
  Information decompositions quantify how the Shannon information about a given random variable is distributed among several other random variables. Various requirements have been proposed that such a decomposition should satisfy, leading to different candidate solutions. Curiously, however, only two of the original requirements that determined the Shannon information have been considered, namely monotonicity and normalization. Two other important properties, continuity and additivity, have not been considered. In this contribution, we focus on the mutual information of two finite variables $Y,Z$ about a third finite variable $S$ and check which of the decompositions satisfy these two properties.  While most of them satisfy continuity, only one of them is both continuous and additive.
\end{abstract}

%\begin{keyword}
%information measures \sep mutual information \sep bivariate information decomposition \sep additivity \sep continuity
%\end{keyword}

\section{Introduction}
\label{sec:introduction}

The fundamental concept of Shannon information is uniquely determined by four simple requirements, {\em continuity, strong additivity, monotonicity}, and a {\em normalization}~\citep{Shannon48:A_Mathematical_Theory_of_Communication_II}.   {We refer to \citet{Csiszar08:Axioms_for_information_measures} for a discussion of axiomatic characterizations.} Continuity implies that small perturbations of the underlying probability distribution have only small effects on the information measure, and this is of course very appealing. Strong additivity refers to the requirement that the chain rule $H(ZY) = H(Y) + H(Z|Y)$ holds.  Similar conditions are also satisfied, mutatis mutandis, for the derived concepts of conditional and mutual information, as well as for other information measures, such as interaction information/co-information \citep{McGill1954,Bell2003} or total correlation/multi-information \citep{Watanabe60:Total_correlation,StudenyVejnarova98:Multiinformation}.

\citet{WilliamsBeer:Nonneg_Decomposition_of_Multiinformation} proposed to decompose the mutual information that several random  variables $Y_{1},\dots,Y_{k}$ have about a target variable $S$ into various components that quantify how much information these variables possess individually, how much they share and how much they need to combine to become useful. That is, one wants to  disentangle how information about  $S$ is distributed over the $Y_{1},\dots,Y_{k}$. Again, various requirements can be imposed, with varying degrees of plausibility, upon such a decomposition. There are several candidate solutions, and not all of them satisfy all those requirements. Curiously, however, previous considerations did not include  continuity and strong additivity.  While \citet{BROJ13:Shared_information} did consider chain rule-type properties, none of the information measures defined within the context of information decompositions satisfies any of these chain rule properties~\citep{RBOJ14:Reconsidering_unique_information}. 

In this contribution, we evaluate which of the various proposed decompositions satisfy \emph{continuity} and \emph{additivity}. Here, additivity (without strong) is required only for independent variables (see Definition~\ref{def:additivity} below). Additivity (together with other properties) may replace strong additivity when defining Shannon information axiomatically (see {\citep{Csiszar08:Axioms_for_information_measures}} for an overview). The importance of additivity is also discussed by~\citet{MatveevPortegies17:Tropical_limits}. 

We consider the case where all random variables are finite, and we restrict to the bivariate case $k=2$.  We think that this simplest possible setting is the most important one to understand conceptually and in practical applications. Already here there are important differences between the measures that have been proposed in the literature.  A bivariate information decomposition consists of three functions $SI$, $UI$ and $CI$ that depend on the joint distribution of three variables $S,Y,Z$, and that satisfy:
\begin{equation}
	\label{eq:infodeco}
	\begin{aligned}
		I(S;YZ) &= \underbrace {{SI}(S;Y,Z)}_{\text{shared}} + \underbrace {CI(S; Y,Z)}_{\text{complementary}}
		+ \underbrace {UI(S;Y\backslash Z)}_{\text{unique ($Y$ wrt $Z$)}} + \underbrace {UI(S;Z\backslash Y)}_{\text{unique ($Z$ wrt $Y$)}}, \\
		I(S;Y) &= SI(S;Y,Z) + UI(S;Y\backslash Z),  \\
		I(S;Z) &= SI(S;Y,Z) + UI(S;Z\backslash Y).
	\end{aligned}
\end{equation}
Hence, $I(S;YZ)$ is decomposed into a shared part that is contained in both $Y$ and $Z$, a complementary (or synergistic) part that is only available from $(Y,Z)$ together, and unique parts contained exclusively in either $Y$ or $Z$.  The different terms $I,SI,CI,UI$ are functions of the joint probability distribution of the three random variables $S,Y,Z$, commonly written with suggestive arguments as in \eqref{eq:infodeco}.

To define a bivariate information decomposition in this sense, it suffices to define either of $SI$, $UI$ or~$CI$. The other functions are then determined from~\eqref{eq:infodeco}. The linear system~\eqref{eq:infodeco} consists of three equations in four unknowns, where the two unknowns $UI(S;Y\backslash Z)$ and $UI(S;Z\backslash Y)$ are related.  Thus, when starting with a function $UI$ to define an information decomposition, the following consistency condition must be satisfied:
\begin{align} \label{eq:consistency}
	I(S;Y)+UI(S;Z \backslash Y)=I(S;Z)+UI(S;Y \backslash Z).
\end{align}
If consistency is not given, one may try to adjust the proposed measure of unique information to enforce consistency using a construction from~\citet{BOJR18:UI_and_Deficiencies} (see Section~\ref{sec:UIconstruction}). 

As mentioned above, several bivariate information decompositions have been proposed (see Section~\ref{sec:measures} for a list).  However, there are still holes in our understanding of the properties of those decompositions that have been proposed so far.  This paper investigates the continuity and additivity properties of some of these decompositions. 

Continuity is understood with respect to the canonical topology of the set of joint distributions of finite variables of fixed sizes.  When $P_{n}$ is a sequence of joint distributions with $P_{n}\to P$, does $SI_{P_{n}}(S;Y,Z)\to SI_{P}(S;Y,Z)$?  Most, but not all, proposed information decompositions are continuous (i.e.\ $SI$, $UI$ and $CI$ are all continuous). If an information decomposition is continuous, one may ask whether it is differentiable, at least at probability distributions of full support. Among the information decompositions that we consider, only the decomposition $I_{\IG}$ \citep{NiuQuinn19:IG_decomposition} is differentiable. Continuity and smoothness are discussed in detail in Section~\ref{sec:continuity}.

The second property that we focus on is additivity, by which we mean that $SI$, $UI$ and $CI$ behave additively when a system can be decomposed into (marginally) independent subsystems (see Definition~\ref{def:additivity} in Section~\ref{sec:additivity}).  This property corresponds to the notion of \emph{extensivity} as used in thermodynamics.
Only the information decomposition $I_{\BROJA}$~\citep{BROJA13:Quantifying_unique_information} in our
list satisfies this property.
A weak form of additivity, the \emph{identity axiom} proposed
by~\citet{HarderSalgePolani2013:Bivariate_redundancy}, is well-studied and is satisfied by other bivariate information decompositions.

\section{Proposed information decompositions}
\label{sec:measures}

We now list the bivariate information decompositions that we want to investigate. The last paragraph mentions further related information measures.
We denote information decompositions by $I$, with sub- or superscripts.  The
corresponding measures $SI$, $UI$ and $CI$ inherit these decorations.

We use the following notation:  $S$, $Y$, $Z$ are random variables with finite state spaces $\Scal$, $\Ycal$, $\Zcal$. The set of all probability distributions on a set $\Scal$ (i.e.\ the probability simplex over $\Scal$) is denoted by $\Pb_{\Scal}$. The joint distribution $P$ of $S,Y,Z$ is then an element of $\Pb_{\Scal\times\Ycal\times\Zcal}$.

\paragraph{$\bullet\;\;\boldsymbol{I_{\min}}$}
\label{sec:Imin}

Together with the information decomposition framework,
\citet{WilliamsBeer:Nonneg_Decomposition_of_Multiinformation} also defined an information
decomposition~$I_{\min}$.  Let 
\begin{align*}
	I(S=s;Y) &= \sum_{y\in\Ycal} P(y|s) \log\frac{P(s|y)}{P(s)}, \\
	I(S=s;Z) &= \sum_{z\in\Zcal} P(z|s) \log\frac{P(s|z)}{P(s)}
\end{align*}
be the \emph{specific information} of the outcome $S=s$ about $Y$ and $Z$, respectively. 
Then
\begin{equation*}
	SI_{\min}(S;Y,Z) = \sum_{s\in\Scal} P(s) \min\big\{ I(S=s;Y), I(S=s;Z) \big\}. 
\end{equation*}
$I_{\min}$ has been criticized, because it assigns relatively large values of shared information,
conflating ``the same amount of information'' with ``the same
information''~\citep{HarderSalgePolani2013:Bivariate_redundancy,GriffithKoch2014:Quantifying_Synergistic_MI}.

\paragraph{$\bullet\;\;\boldsymbol{I_{\MMI}}$}
\label{sec:IMMI}

A related information decomposition is {the minimum mutual information (MMI) decomposition} given by
\begin{equation*}
	SI_{\MMI}(S;Y,Z) = \min\big\{ I(S;Y), I(S;Z) \big\}.
\end{equation*}
Even more severely than $I_{\min}$, this information decomposition conflates ``the same
amount of information'' with ``the same information.''  Still, formally, this definition produces a
valid bivariate information decomposition and thus serves as a useful benchmark.
The axioms imply that $SI(S;Y,Z) \le SI_{\MMI}(S;Y,Z)$ for any other bivariate
information decomposition. For multivariate Gaussian variables, many information decompositions actually agree with~$I_{\MMI}$~\citep{Barrett2014:Gaussian_information_decomposition}.

\paragraph{$\bullet\;\;\boldsymbol{I_{\red}}$}
\label{sec:Ired}

To address the criticism of~$I_{\min}$, \citet{HarderSalgePolani2013:Bivariate_redundancy} introduced a bivariate information decomposition {$I_{\red}$ based on a notion of redundant information} as follows.
Let $\Zcal':=\{z\in\Zcal:P(Z=z)>0\}$ be the support of~$Z$, and let
\begin{equation}\label{eq:defIsearrow}
\begin{split}
	P_{S|y\searrow Z} &= \argmin_{Q\in\conv\{P(S|z) : z\in\Zcal'\}} D(P(S|y)\|Q), \\
	I_{S}(y\searrow Z) & = D(P(S|y)\|P(S)) - D(P(S|y)\|P_{S|y\searrow Z}), \\
	I_{S}(Y\searrow Z) &= \sum\textstyle_{y\in\Ycal}P(y) I_{S}(y\searrow Z).
 \end{split}
\end{equation}
Then
\begin{equation*}
	SI_{\red}(S;Y,Z)
	= \min\big\{
	I_{S}(Y\searrow Z), I_{S}(Z\searrow Y)
	\big\}.
\end{equation*}

\citet{BOJR18:UI_and_Deficiencies} showed that the quantity $I(S;Y)-I_{S}(Y\searrow Z)$ has a decision-theoretic interpretation by way of channel deficiencies \citep{Raginsky2011}.

\paragraph{$\bullet\;\;\boldsymbol{I_{\BROJA}}$}
\label{sec:IBROJA}

Motivated from decision-theoretic considerations, \citet{BROJA13:Quantifying_unique_information} introduced the bivariate information decomposition $I_{\BROJA}$ (eponymously named after the authors in \citep{BROJA13:Quantifying_unique_information}). 
Given $P \in \mathbb{P}_{\Scal\times\Ycal\times\Zcal}$, let $\Delta_P$ denote the set of joint distributions of $(S,Y,Z)$ that have the same marginals on $(S,Y)$ and $(S,Z)$ as $P$. 
Then define the unique information that $Y$ conveys about $S$ with respect to $Z$ as 
\begin{equation*}
	UI_{\BROJA}(S;Y\backslash Z) := \min_{Q \in \Delta_P} I_Q(S;Y|Z),
\end{equation*}
where the subscript $Q$ in $I_Q$ denotes the joint distribution on which the function is computed. 
{Computation of this decomposition was investigated by \citet{8437757}.} 
$I_{\BROJA}$ leads to a concept of synergy that agrees with the synergy measure defined
by~\citet{GriffithKoch2014:Quantifying_Synergistic_MI}.

\paragraph{$\bullet\;\;\boldsymbol{I_{\dep}}$}
\label{sec:Idep}

\citet{JamesEmenheiserCrutchfield18:Idep} define the following bivariate decomposition:
Given the joint distribution $P\in\Pb_{\Scal\times\Ycal\times\Zcal}$ of $(S,Y,Z)$, let
$P_{Y-S-Z} = P(S,Y)P(S,Z) / P(S)$ be the probability distribution in
$\Pb_{\Scal\times\Ycal\times\Zcal}$ that maximizes the entropy among all distributions $Q$ with $Q(S,Y) = P(S,Y)$ and $Q(S,Z) = P(S,Z)$.  Similarly, let $P_{\Delta}$ be the probability distribution in $\Pb_{\Scal\times\Ycal\times\Zcal}$ that maximizes the entropy among all distributions $Q$ with $Q(S,Y) = P(S,Y)$ and $Q(S,Z) = P(S,Z)$ and $Q(Y,Z) = P(Y,Z)$ (unlike for $P_{Y-S-Z}$, {we do not have an explicit formula for $P_\Delta$}). 
Then
\begin{equation*}
	UI_{\dep}(S;Y\setminus Z) = \min\big\{
	I_{P_{Y-S-Z}}(S;Y|Z), I_{P_{\Delta}}(S;Y|Z)
	\big\}.
\end{equation*}
This definition is motivated in terms of a lattice of all sensible marginal constraints when
maximizing the entropy, as in the definition of $P_{Y-S-Z}$ and $P_{\Delta}$
(see~{\citep{JamesEmenheiserCrutchfield18:Idep}} for the details).

\paragraph{$\bullet\;\;\boldsymbol{I^*_\cap}$, $\boldsymbol{I_{\cap}^{\wedge}}$ and $\boldsymbol{I_{\cap}^{\GH}}$}
\label{sec:Icap}

The information decompositions $I_{\cap}^{\wedge}$~\citep{GCJEC14:Common_randomness}, $I_{\cap}^{\GH}$~\citep{GriffithHo15:Quantifying_redundancy} and $I_{\cap}^{*}$~\citep{Kolchinsky19:Multivariate_redundancy_and_synergy}
{are motivated from the notion of common information due to \citet{GacsKoerner73:Common_information} and}
present three different approaches that try to represent the shared information in terms of a random variable~$Q$:
\begin{align}
	SI_{\cap}^{\wedge}(S;Y,Z)
	&= \max \Big\{ I(Q;S) : Q=f(Y)=f'(Z) \text{ a.s.} \Big\}, 
 \label{eq:def-wedge}
 \allowdisplaybreaks \\
	SI_{\cap}^{\GH}(S;Y,Z)
	&= \max \Big\{ I(Q;S) :
	\allowdisplaybreaks[0] 
	I(S;Q|Y) = I(S;Q|Z) = 0 \Big\}, 
 \label{eq:def-GH}
 \allowdisplaybreaks \\
	\pagebreak[0]
	SI^{*}_{\cap}(S;Y,Z)
	&= \max \Big\{ I(Q;S) :
	P(s,q) = \sum_{y}P(s,y) \lambda_{q|y}
	= \sum_{z}P(s,z) \lambda'_{q|z}
	\Big\}, 
 \label{eq:def-SI*}
\end{align}
where the optimization runs over all pairs of (deterministic) functions $f,f'$ (for $SI_{\cap}^{\wedge}$), 
all joint distributions of four random variables $S,X,Y,Q$ that extend the joint distribution of~$S,X,Y$ (for $SI_{\cap}^{\GH}$),  
and all pairs of stochastic matrices $\lambda_{q|y},\lambda'_{q|z}$ (for $SI_{\cap}^{*}$), respectively. 
One can show that $SI_{\cap}^{\wedge}(S;Y,Z)\le SI_{\cap}^{\GH}(S;Y,Z)\le SI_{\cap}^{*}(S;Y,Z)$~\citep{Kolchinsky19:Multivariate_redundancy_and_synergy}.

The $I_{\cap}^{*}$-decomposition draws motivation from considerations of channel preorders, in a similar spirit as~\citet{BOJR18:UI_and_Deficiencies}, and it is related to ideas from~\citet{BR13:Blackwell_relation_and_zonotopes}.
\citet{Kolchinsky19:Multivariate_redundancy_and_synergy} shows that there is an analogy between $I_{\cap}^{*}$
and~$I_{\BROJA}$.

\paragraph{$\bullet\;\;\boldsymbol{I_{\IG}}$}
\label{sec:IIG}

\citet{NiuQuinn19:IG_decomposition} presented a bivariate information decomposition
$I_{\IG}$ based on information geometric ideas.  While their construction is very elegant, it only
works for joint distributions $P$ of full support (i.e.\ $P(s,y,z) > 0$ for all $s,y,z$).  It is
unknown whether it can be extended meaningfully to all joint distributions.  Numerical evidence
exists that a unique continuous extension is possible at least to some joint distributions with
restricted support (see examples {in~\citep{NiuQuinn19:IG_decomposition}}).

For any $t\in\Rb$, consider the joint distribution
\begin{equation*}
	P^{(t)}(s,y,z) = \frac{1}{c_{t}} P^{t}_{S-Y-Z}(s,y,z) P^{1-t}_{S-Z-Y}(s,y,z)
	= \frac{1}{c_{t}} P(y,z) P(s|y)^{t} P(s|z)^{1-t},
\end{equation*}
where $c_{t}$ is a normalizing constant, and let
\begin{equation*}
	P^{*} = \argmin_{t\in\Rb} D(P\|P^{(t)}).
\end{equation*}
Then
\begin{align*}
	{CI_{\IG}}(S;Y,Z) &= D(P\|P^{*}), 
	\\
	UI_{\IG}(S;Y\setminus Z) &= D(P^{*}\|P_{S-Z-Y}).
\end{align*}
{An interesting aspect about these definitions is that, by the Generalized Pythagorean Theorem (see \citep{10.5555/3239831}), $D(P\|P_{S-Z-Y}) = D(P\|P^*) + D(P^*\|P_{S-Z-Y})$. }

\paragraph{$\bullet\;$ \textbf{The UI construction}}
\label{sec:UIconstruction}

Given an information measure that captures some aspect of unique information but that fails to satisfy the consistency condition~\eqref{eq:consistency}, one may construct a corresponding bivariate information decomposition as follows:
\begin{lemma}
	\label{lem:UIconstruction}
	Let $\delta:\Pb_{\Scal\times\Ycal\times\Zcal}\to\Rb$ be a non-negative function that satisfies
	\begin{equation*}
		\delta(S;Y\setminus Z) \le \min\{ I(S;Y), I(S;Y|Z) \}.
	\end{equation*}
	Then a bivariate information decomposition is given by
	\begin{equation*}
		\begin{aligned}
			UI_{\delta}(S;Y\setminus Z) &= \max\big\{ \delta(S;Y\setminus Z), %\\&\qquad\qquad
			\delta(S;Z\setminus Y) + I(S;Y) - I(S;Z) \big\}, \\
			UI_{\delta}(S;Z\setminus Y) &= \max\big\{ \delta(S;Z\setminus Y), %\\ &\qquad\qquad 
			\delta(S;Y\setminus Z) + I(S;Z) - I(S;Y) \big\}, \\
			SI_{\delta}(S;Z,Y) &= \min\big\{ I(S;Y) - \delta(S;Y\setminus Z), %\\ &\qquad\qquad 
			I(S;Z) - \delta(S;Z\setminus Y) \big\}, \\
			CI_{\delta}(S;Z,Y) &= \min\big\{ I(S;Y|Z) - \delta(S;Y\setminus Z), %\\ &\qquad\qquad
			I(S;Z|Y) - \delta(S;Z\setminus Y) \big\}.
		\end{aligned}
	\end{equation*}
\end{lemma}
\begin{proof}
	The proof follows just as the proof of \citep[Proposition~9]{banerjee2020variational}.
\end{proof}
We call the construction of Lemma~\ref{lem:UIconstruction} the \emph{UI
	construction}.  The unique information $UI_{\delta}$ returned by the UI construction is the
smallest $UI$-function of any bivariate information decomposition with $UI\ge\delta$.

As~\citet{BOJR18:UI_and_Deficiencies} show, the decomposition $I_{\red}$ is an example of this
construction.
As another example, as \citet{BOJR18:UI_and_Deficiencies} and~\citet{RBOJ19:UI_and_Secret_Key_Decompositions} suggested, the UI
construction can be used to obtain bivariate information decompositions from the one- or two-way secret key rates and related information functions that have been defined as bounds on the secret key rates, such as the intrinsic information~\citep{MaurerWolf97:intrinsic_conditional_MI}, the
reduced intrinsic information~\citep{RennerW03}, or the minimum intrinsic information~\citep{GohariAnantharam10:Information_theoretic_key_agreement_I}.

\paragraph{$\bullet\;\;$ \textbf{Other decompositions}}
\label{sec:Iother}

Several other measures have been proposed that are motivated by the framework of~\citet{WilliamsBeer:Nonneg_Decomposition_of_Multiinformation} but that leave the framework.  \citet{Ince17:Iccs} defines a decomposition $I_{\ccs}$, which satisfies~\eqref{eq:infodeco}, but in which $SI_{\ccs}$, $UI_{\ccs}$ and $CI_{\ccs}$ may take negative values.  The SPAM decomposition of {\citet{FinnLizier18:Pointwise_PI_Using_SPAM}} consists of non-negative information measures that decompose the mutual information, but this decomposition has a different structure, with alternating signs and twice as many terms. Both approaches construct ``pointwise'' decompositions, in the sense that $SI$, $UI$ and $CI$ can be naturally expressed as expectations, in a similar way that entropy and mutual information can be written as expectations (see~{\citep{FinnLizier18:Pointwise_PI_Using_SPAM}} for details). 
{Recent works have proposed other decompositions based on a different lattice \citep{Ay2019InformationDB} or singling out features of the target variable \citep{https://doi.org/10.13140/rg.2.2.27408.33289}.} 

Since these measures do not lie in our direct focus, we omit their definitions.  Nevertheless, one can ask the same questions: Are the corresponding information measures continuous, and are they additive?
For the constructions in~{\citep{FinnLizier18:Pointwise_PI_Using_SPAM}}, both continuity and additivity (as a consequence of a chain rule) are actually postulated. {The decomposition in \citep{Ay2019InformationDB} is additive.} 
On the other hand, $I_{\ccs}$ is neither continuous (as can be seen from its definition) nor additive (since it does not satisfy the identity property).

\section{Continuity}
\label{sec:continuity}

Most of the information decompositions that we consider are continuous.
Moreover, the UI construction preserves continuity: if $\delta$ is continuous, then $UI_{\delta}$ is continuous. 

The notable exceptions to continuity are $I_{\red}$ and the $I_{\cap}$ decompositions (see Lemmas~\ref{lem:red-non-cont} and~\ref{lem:cap-non-cont} below). For $SI_{\red}$, this is due to its definition in terms of conditional probabilities.  Thus, $SI_{\red}$ is continuous when restricted to probability distributions of full support.  For $SI^{*}_{\cap}$,
discontinuities also appear for sequences $P_{n}\to P$ where the support does not change.

For the $SI_{\IG}$ information decomposition, one should keep in mind that it is only defined for probability distributions with full support.  It is currently unknown whether it can be continuously extended to {all} probability distributions. 

Clearly, continuity is a desirable property, but is it essential?  A discontinuous information measure might still be useful, if the discontinuity is not too severe.  For example, the G{\'a}cs-K{\"o}rner common information $C(Y\wedge Z)$ \citep{GacsKoerner73:Common_information} is an information measure that vanishes except on a set of measure zero {(certain distributions that do not have full support)}. 
Clearly, such an information measure is difficult to estimate.
The $I_{\cap}$ decompositions are related to~$C(Y\wedge Z)$, and so their discontinuity is almost as severe (see Lemma~\ref{lem:cap-non-cont}).
{Similarly,} the $I_{\red}$-decomposition is continuous at distributions of full support.
If the discontinuity is well-behaved and well understood, then such a decomposition may still be useful for certain applications.
Still, a discontinuous information decomposition challenges the intuition, and any discontinuity must be interpreted (such as the discontinuity of $C(Y\wedge Z)$ can be explained and interpreted~\citep{GacsKoerner73:Common_information}).

If an information decomposition is continuous, one may ask whether it is differentiable, at least at probability distributions of full support.  For almost all information decompositions that we consider, the answer is no.  This is easy to see for those information decompositions that involve a minimum of finitely many smooth functions ($SI_{\min}$, $SI_{\MMI}$, $SI_{\red}$, $SI_{\dep}$).  For $SI_{\BROJA}$, we refer to~\citet{RauhSchuenemann20:Properties_of_UI}. 
Only $SI_{\IG}$ is differentiable for distributions of full support\footnote{Personal communication with the authors \citet{NiuQuinn19:IG_decomposition}.}. 

\begin{lemma}
	\label{lem:red-non-cont}
	$SI_{\red}$ is not continuous.
\end{lemma}
\begin{proof}
	$I_{S}(Y\searrow Z)$ and $I_{S}(Z\searrow Y)$ are defined in terms of conditional probability
	$P(S|Y=y)$ and $P(S|Z=z)$, which are only defined for those $y,z$ with $P(Y=y)>0$ and $P(Z=z)>0$.
	Therefore, $I_{S}(Y\searrow Z)$ and $I_{S}(Z\searrow Y)$ are discontinuous when probabilities tend
	to zero.  A concrete example is given below.
\end{proof}

\begin{example}[$SI_{\red}$ is not continuous]
	\label{sec:noncont-Ired}
	For $0\le a \le 1$, suppose that the joint distribution of $S,Y,Z$ has the following marginal distributions:
	\begin{center}
 \renewcommand*{\arraystretch}{1.1}
		\begin{tabular}{cccc}
			\toprule
			$s$ & $y$ & $P_{a}(s,y)$ \\
			\midrule
			1 & 0 & $\frac{a}{2}$ \\
			1 & 1 & $\frac{1}{2} - \frac{a}{2}$ \\
			0 & 1 & $\frac{1}{4}$ \\
			0 & 2 & $\frac{1}{4}$ \\
			\bottomrule
		\end{tabular}
		\hfil
		\begin{tabular}{cccc}
			\toprule
			$s$ & $z$ & $P_{a}(s,z)$ \\
			\midrule
			0 & 0 & $\frac{a}{2}$ \\
			0 & 1 & $\frac{1}{2} - \frac{a}{2}$ \\
			1 & 1 & $\frac{1}{4}$ \\
			1 & 2 & $\frac{1}{4}$ \\
			\bottomrule
		\end{tabular}.
	\end{center}
	Observe the symmetry of $Y$, $Z$. 
	For $a>0$, the conditional distributions of $S$ given $Y$ and $Z$ are, respectively: 
	\begin{center}
  \renewcommand*{\arraystretch}{1.3}
		\begin{tabular}{cc}
			\toprule
			$y$ & $P_{a}(S|y)$ \\
			\midrule
			0 & $\left({0},{1}\right)$ \\ 
			1 & $\left({\frac{1}{3-2a}},{\frac{2-2a}{3-2a}}\right)$ \\
			2 & $\left({1},{0}\right)$ \\
			\bottomrule
		\end{tabular}
		\hfil
		and
		\hfil
		\begin{tabular}{cc}
			\toprule
			$z$ & $P_{a}(S|z)$ \\
			\midrule
			0 & $\left({1},{0}\right)$ \\
			1 & $\left({\frac{2 - 2a}{3 - 2a}},{\frac{1}{3 - 2a}}\right)$ \\
			2 & $\left({0},{1}\right)$ \\
			\bottomrule
		\end{tabular}.
	\end{center}
	Therefore, $I_{S}(Y\searrow Z) = I(S;Y) = I(S;Z) = I_{S}(Z\searrow Y)$. {The first equality follows from the definition of $I_{S}(Y\searrow Z)$ in \eqref{eq:defIsearrow}, noting that $\conv\{P(S|z)\colon P(z)>0\}$ includes all probability distributions on $\{0,1\}$ and hence $D(P(S|y)\|P_{S|y\searrow Z})=0$. 
    The second equality holds because the marginal distribution of $(S,Y)$ is equal to that of $(S,Z)$ up to relabeling of the states of $S$. The third equality follows by similar considerations as the first.}
	
	For $a=0$, the conditional distributions $P(S|Y=0)$ and $P(S|Z=0)$ are not defined. In this case we have $I_{S}(Y\searrow Z) = I_{S}(Z\searrow Y) < I(S;Y) = I(S;Z)$. {As before, the two equalities hold because the marginal distribution of $(S,Y)$ is equal to that of $(S,Z)$ up to relabeling of the states of $S$. 
    The inequality holds because now $\conv\{P(S|z)\colon P(z)>0\}$ does not include all probability distributions on $\{0,1\}$ and $D(P(S|y)\|P_{S|y\searrow Z})>0$ for $y=2$.} 
 In total,
	\begin{gather*}
		\lim_{a\to 0+} SI_{\red}(S;Y,Z) = \lim_{a\to 0+} I(S;Y) > I_{S}(Y\searrow Z) = SI_{\red}(S;Y,Z).
	\end{gather*}
\end{example}

\begin{lemma}
	\label{lem:cap-non-cont}
	$SI_{\cap}^{*}$, $I_{\cap}^{\wedge}$ and $I_{\cap}^{\GH}$ are not continuous.
\end{lemma}
\begin{proof}
	By~\citet[Theorem~6]{Kolchinsky19:Multivariate_redundancy_and_synergy}, 
    for all
	three measures, $SI_{\cap}(YZ;Y,Z)$ equals the G{\'a}cs-K{\"o}rner common information~$C(Y\wedge Z)$
	\citep{GacsKoerner73:Common_information}, which is not continuous.
\end{proof}
A concrete example is given below. 

\begin{example}[$SI_{\cap}^{*}$ is not continuous]
	Suppose that the joint distribution of $S,Y,Z$ has the following marginal distributions, for $-1\le a \le 1$: 
	\begin{center}
\renewcommand*{\arraystretch}{1.1} 
		\begin{tabular}{cccc}
			\toprule
			$s$ & $y$ & $P_{a}(s,y)$ \\
			\midrule
			0 & 0 & $\frac{1}{3}$ \\
			1 & 0 & $\frac{1}{6} - \frac{a}{6}$ \\
			1 & 1 & $\frac{1}{6} + \frac{a}{6}$ \\
			2 & 1 & $\frac{1}{3}$ \\
			\bottomrule
		\end{tabular}
		\hfil
		\begin{tabular}{cccc}
			\toprule
			$s$ & $z$ & $P_{a}(s,z)$ \\
			\midrule
			0 & 0 & $\frac{1}{3}$ \\
			1 & 0 & $\frac{1}{6}$ \\
			1 & 1 & $\frac{1}{6}$ \\
			2 & 1 & $\frac{1}{3}$ \\
			\bottomrule
		\end{tabular}. 
	\end{center}
 {Recall the definition of $SI_{\cap}^{*}(S;Y,Z)$ in \eqref{eq:def-SI*}.} 
	For $a=0$, the marginal distributions of the pairs $(S,Y)$ and $(S,Z)$ are identical, whence $SI_{\cap}^{*}(S;Y,Z) = I(S;Y) = I(S;Z)$.
	
	Now let $a\neq 0$. According to the definition of $SI^{*}_{\cap}$, we need to find stochastic
	matrices $\lambda_{q|y},\lambda'_{q|z}$ that satisfy the condition
	\begin{equation}
		\label{eq:consistency_Icap}
		P(s,q) = \sum_{y}P(s,y) \lambda_{q|y} = \sum_{z}P(s,z) \lambda'_{q|z} . 
	\end{equation}
	For $s=0$ and $s=2$, this condition implies $\lambda_{q|0} = \lambda'_{q|0}$ and
	$\lambda_{q|1} = \lambda'_{q|1}$.
	For $s=1$, the same condition gives
	\begin{equation*}
		a (\lambda_{q|1} - \lambda_{q|0}) = 0.
	\end{equation*}
	In the case $a\neq 0$, this implies that $\lambda_{q|1} = \lambda_{q|0}$ and {thus} that $Q$ is independent of {$Y$ and}~$S$. 
    Therefore, {$I(Q;S)=0$} and $SI_{\cap}^{*}(S;Y,Z) = 0$ for {$a\neq0$}. 
\end{example}

\paragraph{\textbf{Asymptotic continuity and locking}}
We discuss two further related properties, namely, \emph{asymptotic continuity} and \emph{locking} \citep{BOJR18:UI_and_Deficiencies,RBOJ19:UI_and_Secret_Key_Decompositions}, {which we explain shortly below}. {$I_{\BROJA}$ is asymptotically continuous and does not exhibit locking.  It is not known if other information decompositions satisfy these properties.}

Operational quantities in information theory such as channel capacities and compression rates are usually defined in the spirit of \citet{Shannon48:A_Mathematical_Theory_of_Communication_II} - in the asymptotic regime of many independent uses of the channel or many independent realizations of the underlying source distribution. In the asymptotic regime, real-valued functionals of distributions that are asymptotically continuous are especially useful as they often provide lower or upper bounds for operational quantities of interest \citep{cerf2002multipartite,BOJR18:UI_and_Deficiencies,chitambar2019quantum}. 

Asymptotic continuity is  a stronger notion of continuity that considers convergence relative to the dimension {of} the underlying state space \citep{synak2006asymptotic,chitambar2019quantum,fannes1973continuity,winter2016tight}. Concretely, a function $f$ is said to be asymptotically continuous if
$$|f(P)-f(P')| \le C\epsilon \log |{\mathcal{S}}| + \zeta(\epsilon)$$
for all joint distributions $P,P'\in {\mathbb{P}_{\mathcal{S}}}$,
where $C$ is some constant, $\epsilon=\tfrac{1}{2}\|P-P'\|_1$, and $\zeta:[0,1]\to \Rb_+$ is any continuous function converging to zero as $\epsilon \to 0$ \citep{chitambar2019quantum}.

As an example, entropy is asymptotically continuous (see, e.g., \citep[Lemma 2.7]{csiszar2011information}):
For any~$P,P'\in\mathbb{P}_{\Scal}$, if $\tfrac{1}{2}\|P-P'\|_1\le \epsilon$, then 
$$|H_P(S)-H_{P'}(S)| \le \epsilon \log |\mathcal{S}| + h(\epsilon),$$
where $h(\cdot)$ is the binary entropy function, $h(p)=-p\log p -(1-p)\log (1-p)$ for $p\in (0,1)$ and $h(0)=h(1)=0$. Likewise, the conditional mutual information satisfies asymptotic continuity in the following sense  \citep{RennerW03,christandl2004squashed}: For any~$P,P'\in\mathbb{P}_{\Scal\times\Ycal\times\Zcal}$, if $\tfrac{1}{2}\|P-P'\|_1\le \epsilon$ then $$|I_{P}(S;Y|Z)-I_{P'}(S;Y|Z)|\le \epsilon\log\min\{|\Scal|,|\Ycal|\}+2h(\epsilon).$$
Note that the right-hand side of the above inequality does not depend explicitly on the cardinality of $Z$.

As~\citet{BOJR18:UI_and_Deficiencies} show, $UI_{\BROJA}$ is asymptotically continuous:
\begin{lemma}\label{thm:AC}
	For any~$P,P'\in \mathbb{P}_{\Scal\times\Ycal\times\Zcal}$, and~$\epsilon\in [0,1]$, if~$\|P-P'\|_1= \epsilon$, then
	$$|{UI}_{P}(S;Y\backslash Z) - {UI}_{P'}(S;Y\backslash Z)| \le \tfrac{5}{2} \epsilon \log\min\{|\Scal|, |\Ycal|\}+\zeta(\epsilon)$$ 
	for some bounded, continuous function $\zeta:[0,1]\to \Rb_+$ 
	that converges to zero as $\epsilon \to 0$.
\end{lemma}

Locking is motivated from the following property of the {conditional mutual information} \citep{RennerW03,christandl2007unifying}: 
For arbitrary discrete random variables $S$, $Y$, $Z$, and $U$, 
\begin{align}\label{eq:SKlockingproperty}
	I (S ; Y |Z U ) \ge I (S ; Y |Z ) - H(U).
\end{align}
The {conditional mutual information} does not exhibit ``locking'' in the sense 
that any additional side information $U$ accessible to $Z$ cannot reduce the {conditional mutual information} by more than the entropy of the side information. 

As~\citet{RBOJ19:UI_and_Secret_Key_Decompositions} show, 
$UI_{\BROJA}$ does not exhibit locking:
\begin{lemma}
	For jointly distributed random variables $(S,Y,Z,U)$,
	\begin{align}\label{eq:UIlockingproperty}
		UI_{\BROJA}(S ; Y |Z U ) \ge UI_{\BROJA}(S ; Y |Z ) - H(U).
	\end{align}
\end{lemma}
This property is useful, for example, in a cryptographic context \citep{RBOJ19:UI_and_Secret_Key_Decompositions} where it ensures 
that the unique information that $Y$ has about $S$ w.r.t. an adversary $Z$ cannot ``unlock'', i.e., drop by an arbitrarily large amount on giving away a bit of information to $Z$.

\section{Additivity}
\label{sec:additivity}

\begin{definition}
	\label{def:additivity}
	An information measure $I(X_{1},\dots,X_{n})$ (i.e. a function of the joint distribution of $n$ random variables) is \emph{additive} if and only if the following holds:
	If $(X_{1},\dots,X_{n})$ %are 
 {is} independent of $(Y_{1},\dots,Y_{n})$, then
	\begin{equation*}
		I(X_{1}Y_{1}, X_{2}Y_{2},\dots, X_{n}Y_{n})
		= I(X_{1},\dots,X_{n}) + I(Y_{1},\dots,Y_{n}).
	\end{equation*}
	The information measure is \emph{superadditive}, if, under the same assumptions,
	\begin{equation*}
		I(X_{1}Y_{1}, X_{2}Y_{2},\dots, X_{n}Y_{n})
		\ge I(X_{1},\dots,X_{n}) + I(Y_{1},\dots,Y_{n}).
	\end{equation*}
\end{definition}

The $I_{\BROJA}$ decomposition is additive:
\begin{lemma}
	$I_{\BROJA}$ is additive.
\end{lemma}
\begin{proof}
	This is~\citep[Lemma~19]{BROJA13:Quantifying_unique_information}.
\end{proof}

{The information decompositions motivated from the G{\'a}cs-K{\"o}rner common information as defined in \eqref{eq:def-wedge}, \eqref{eq:def-GH} and
\eqref{eq:def-SI*} are additive (Theorem~\ref{thm:Icap-additivity}).}
All other information decompositions that we consider are not additive.  
However, in all information
decompositions that we consider, $SI$ is superadditive and $UI$ is subadditive
(Theorem~\ref{thm:superadditive}).

Again, additivity is a desirable property, but is it essential?  As in the case of continuity, we argue that non-additivity challenges the intuition, and any non-additivity must be interpreted.  Why is it plausible that the shared information contained in two independent pairs is more than the sum of the individual shared informations, and how can one explain that the unique information is subadditive?

A related weaker property is additivity under i.i.d.\ sequences, i.e.\ when, in the definition of additivity, {the vectors $(X_{1},\ldots,X_{n})$ and $(Y_{1},\ldots,Y_{n})$ are identically distributed.}  One can show that $I_{\red}$, $I_{\MMI}$, $I_{\dep}$ and $I_{\IG}$ (and, of course, $I_{\BROJA}$) are additive under i.i.d.\ sequences, but not $I_{\min}$. The $UI$ construction gives additivity of $I_{\delta}$ under i.i.d.\ sequences if $\delta$ is additive under i.i.d.\ sequences.  The proof of these statements is similar to the proof for additivity (given below) and omitted. For the $I_{\cap}$ decompositions, it is not as easy to see, and so we currently do not know whether additivity under i.i.d.\ sequences holds. 

\begin{lemma}
	\label{lem:min-superadditive}
	\begin{enumerate}[leftmargin=*]
		\item If $I_{1}$ and $I_{2}$ are superadditive, then $\min\{I_{1},I_{2}\}$ is superadditive.
		\item If, in addition, there exist distributions $P,Q$ with $I_{1}(P) < I_{2}(P)$ and $I_{1}(Q) > I_{2}(Q)$, then $\min\{I_{1},I_{2}\}$ is not additive.
	\end{enumerate}
\end{lemma}
\begin{proof}
	\begin{enumerate}[leftmargin=*]
		\item 
		With $X_{1},\dots,X_{n},Y_{1},\dots,Y_{n}$ as in the definition of superadditivity,
		\begin{align*}
			\min&\big\{I_{1}(X_{1}Y_{1}, X_{2}Y_{2},\dots, X_{n}Y_{n}), I_{2}(X_{1}Y_{1}, X_{2}Y_{2},\dots, X_{n}Y_{n}) \big\}
			\\ &
			\ge \min\big\{I_{1}(X_{1},\dots,X_{n}) + I_{1}(Y_{1},\dots,Y_{n}), %\\&\qquad\qquad\quad 
			I_{2}(X_{1},\dots,X_{n}) + I_{2}(Y_{1},\dots,Y_{n}) \big\}
			\\ &
			\ge \min\big\{I_{1}(X_{1},\dots,X_{n}), I_{2}(X_{1},\dots,X_{n})\big\}
			%\\ &\quad
			+ \min\big\{I_{1}(Y_{1},\dots,Y_{n}), I_{2}(Y_{1},\dots,Y_{n}) \big\}.
		\end{align*}
		\item 
		In this inequality, if $X_{1},\dots,X_{n}\sim P$ and $Y_{1},\dots,Y_{n}\sim Q$, then the right
		hand side equals $I_{1}(X_{1},\dots,X_{n}) + I_{2}(Y_{1},\dots,Y_{n})$, which makes the inequality
		strict.\qedhere
	\end{enumerate}  
\end{proof}
As a consequence:
\begin{lemma}
	\label{lem:UIconstruction-additivity}
	If $\delta$ is subadditive, then $UI_{\delta}$ is subadditive, $SI_{\delta}$ is superadditive, but
	neither is additive.
\end{lemma}
\begin{theorem}
\label{thm:superadditive}
The shared information measures    $SI_{\min}$, $SI_{\MMI}$, $SI_{\red}$, $SI_{\dep}$, and $SI_{\IG}$ are superadditive, but not additive.
\end{theorem}
\begin{proof}
	For $I_{\MMI}$, the claim follows directly from Lemma~\ref{lem:min-superadditive}.
	The same is true for $I_{\dep}$, since $I_{P_{Y-S-Z}}(S;Y|Z)$ and $I_{P_{\Delta}}(S;Y|Z)$ are additive,
	and also for $I_{\red}$, since $I_{S}(Y\searrow Z)$ and $I_{S}(Z\searrow Y)$ are superadditive.
	For $I_{\min}$, the same argument as in the proof of Lemma~\ref{lem:min-superadditive} applies, since the specific information is additive, in the sense that
	\begin{equation*}
		I(S_{1}S_{2}=s_{1}s_{2}; Y_{1}Y_{2})
		= I(S_{1}=s_{1}; Y_{1}) + I(S_{2}=s_{2}; Y_{2}).
	\end{equation*}
	
	Next, consider {$I_{\IG}$}. 
	For $i=1,2$
	\begin{equation*}
		P_{i}^{(t)}(s_{i},y_{i},z_{i})
		= \frac{1}{c_{i,t}} P(y_{i},z_{i}) P(s_{i}|y_{i})^{t} P(s_{i}|z_{i})^{1-t}.
	\end{equation*}
	Then
	\begin{equation*}
		P^{(t)}(s_{1}s_{2},y_{1}y_{2},z_{1}z_{2})
		= P_{1}^{(t)}(s_{1},y_{1},z_{1})P_{2}^{(t)}(s_{2},y_{2},z_{2})
	\end{equation*}
	and
	\begin{equation*}
		D(P\|P^{(t)}) = D(P_{1}\|P_{1}^{(t)}) + D(P_{2}\|P_{2}^{(t)}),
	\end{equation*}
	where $P_{i}$ denotes the marginal distribution of $S_{i},Y_{i},Z_{i}$ for $i=1,2$.
	It follows that
	\begin{align*}
		SI(S_{1}S_{2}; Y_{1}Y_{2},Z_{1}Z_{2})
		& = \min_{t\in\Rb}D(P\|P^{(t)}) \\ 
		& \ge \min_{t\in\Rb} D(P_{1}\|P_{1}^{(t)}) + \min_{t\in\Rb} D(P_{2}\|P_{2}^{(t)})\\ 
		& = SI(S_{1}; Y_{1},Z_{1}) + SI(S_{2}; Y_{2},Z_{2}).
	\end{align*}
	If $\argmin_{t\in\Rb} D(P_{1}\|P_{1}^{(t)}) \neq \argmin_{t\in\Rb} D(P_{2}\|P_{2}^{(t)})$, then
	strict inequality holds.
\end{proof}

\begin{theorem}
	\label{thm:Icap-additivity}
 {$I_{\cap}^{\wedge}$, $I_{\cap}^{\GH}$ and $I_{\cap}^{*}$} are additive.
\end{theorem}

\begin{proof} 
{In the following, to slightly simplify notation, we write $I_{\wedge}$, $I_{\GH}$ and $I_{*}$ for the decompositions defined in \eqref{eq:def-wedge}, \eqref{eq:def-GH} and
\eqref{eq:def-SI*}, respectively.} 
\begin{itemize}
\item First, consider~$I_{\wedge}$.  As \citet{GCJEC14:Common_randomness} show, 
	$SI_{\wedge}(S_{1}S_{2};Y_{1}Y_{2},Z_{1}Z_{2}) = I(S_{1}S_{2};Q)$, where $Q$ is the \emph{common random variable}
	\citep{GacsKoerner73:Common_information}, which satisfies $Q=Q_{1}Q_{2}$, where $Q_{j}$ is the common random variable
	of $Y_{j}$ and $Z_{j}$.  Therefore,
	\begin{align*}
		SI_{\wedge}(S_{1}S_{2};Y_{1}Y_{2},Z_{1}Z_{2}) 
		=& I(S_{1}S_{2};Q_{1}Q_{2})
		\\
		=& I(S_{1};Q_{1}) + I(S_{2};Q_{2})\\
		=& SI_{\wedge}(S_{1};Y_{1},Z_{1}) + SI_{\wedge}(S_{2};Y_{2},Z_{2}).
	\end{align*}
\item 
	To see that $SI_{\GH}$ is superadditive, suppose that
	$SI_{\GH}(S_{j};Y_{j},Z_{j}) = I(Q_{j};S_{j})$. 
	The joint distribution of $S_{1},S_{2},Q_{1},Q_{2}$ defined by
		$P(s_{1}s_{2}q_{1}q_{2}) = P(s_{1}q_{1})P(s_{2}q_{2})$
	is feasible for the optimization problem in the definition of
	$SI_{\GH}(S_{1}S_{2};Y_{1}Y_{2},Z_{1}Z_{2})$.  Therefore,
	\begin{align*}
		SI_{\GH}(S_{1}S_{2};Y_{1}Y_{2},Z_{1}Z_{2}) 
		&\ge I(S_{1}S_{2};Q_{1}Q_{2}) \\
		& = I(S_{1};Q_{1}) + I(S_{2};Q_{2}) \\
		& = SI_{\GH}(S_{1};Y_{1},Z_{1}) + SI_{\GH}(S_{2};Y_{2},Z_{2}).
	\end{align*}
	To prove subadditivity, let $Q$ be as in the definition of $SI_{\GH}(S_{1}S_{2};Y_{1}Y_{2},Z_{1}Z_{2})$ {in \eqref{eq:def-GH}}, with
 {($S_{1},Y_{1},Z_{1}$), ($S_{2},Y_{2},Z_{2}$)} as in Definition~\ref{def:additivity}. 
	The chain rule implies 
 $$
 I(S_{1}S_{2};Q) = I(S_{1};Q) + I(S_{2};Q|S_{1}), 
 $$
 where
	$I(S_{2};Q|S_{1}) = \sum_{s_{1}} P(S_{1}=s_{1})I(S_{2};Q|S_{1}=s_{1})$.  
    Choose an element $s_{1}^{*}$ in $\argmax_{s_{1}}I(S_{2};Q|S_{1}=s_{1})$. 
	Construct two random variables $Q_{1},Q_{2}$ as follows:
	$Q_{1}$ is independent of $S_{2},Y_{2},Z_{2}$ and satisfies $P(Q_{1}|S_{1},Y_{1},Z_{1}) = P(Q|S_{1},Y_{1},Z_{1})$.
	$Q_{2}$ is independent of $S_{1},Y_{1},Z_{1}$ and satisfies $P(Q_{2}|S_{2},Y_{2},Z_{2}) = P(Q|S_{2},Y_{2},Z_{2},S_{1}=s_{1}^{*})$.
	By construction, $Q_{1}Q_{2}$ is independent of $S_{1}S_{2}$ given $Y_{1}Y_{2}$, and $Q_{1}Q_{2}$ is independent of
	$S_{1}S_{2}$ given $Z_{1}Z_{2}$.
	The statement follows from
	\begin{align*}
		{SI_{\GH}}&{(S_{1};Y_{1},Z_{1}) + SI_{\GH}(S_{2};Y_{2},Z_{2}) }\\
		& \ge 
        I(S_{1};Q_{1}) + I(S_{2};Q_{2}) 
		= I(S_{1};Q) + I(S_{2};Q|S_{1}=s_{1}) \\
		& \ge I(S_{1};Q) + I(S_{2};Q|S_{1}) 
		= I(S_{1}S_{2};Q)
		= SI_{\GH}(S_{1}S_{2};Y_{1}Y_{2},Z_{1}Z_{2}). 
	\end{align*}
 \item 	
  {The proof of superadditivity for $I_*$ follows line by line the proof for $I_{\GH}$.} 
	To prove subadditivity for $I_{*}$, we claim that for all random variables $S,Y,Z$ there exist random variables 
	{$S'=f_S(S,Y,Z)$, $Y'=f_Y(S,Y,Z)$, $Z'=f_Z(S,Y,Z)$} with $P(S,Y) = P(S',Y')$, $P(S,Z) = P(S',Z')$ and $I_{*}(S;Y,Z) = I_{\GH}(S';Y',Z')$. 
 This correspondence can be chosen such that
 {
 \begin{align*}
 (S_{1}S_{2})'=f_S(S_1S_2,Y_1Y_2,Z_1Z_2)&=(f_S(S_1,Y_1,Z_1),f_S(S_2,Y_2,Z_2)) = S_{1}'S_{2}', \\
 (Y_{1}Y_{2})'=f_Y(S_1S_2,Y_1Y_2,Z_1Z_2)&=(f_Y(S_1,Y_1,Z_1),f_Y(S_2,Y_2,Z_2)) = Y_{1}'Y_{2}', \\
 (Z_{1}Z_{2})'=f_Z(S_1S_2,Y_1Y_2,Z_1Z_2)&=(f_Z(S_1,Y_1,Z_1),f_Z(S_2,Y_2,Z_2)) = Z_{1}'Z_{2}',
 \end{align*}}%
 where 
 $S_{1}'Y_{1}'Z_{1}'$ is independent of $S_{2}'Y_{2}'Z_{2}'$. 
 Thus, 
	\begin{align*}
		SI_{*}(S_{1}S_{2};Y_{1}Y_{2},Z_{1}Z_{2}) %\\
		& = SI_{\GH}\big((S_{1}S_{2})';(Y_{1}Y_{2})',(Z_{1}Z_{2})'\big) \\
		& = SI_{\GH}(S'_{1};Y'_{1},Z'_{1}) + SI_{\GH}(S'_{2};Y'_{2},Z'_{2}) \\
		& \le SI_{*}(S'_{1};Y'_{1},Z'_{1}) + SI_{*}(S'_{2};Y'_{2},Z'_{2}).
	\end{align*}
	To prove the claim, suppose that $SI_{*}(S;Y,Z) = I(S;Q)$, with $Q$ as in the definition of $SI_{*}$ {in \eqref{eq:def-SI*}}. 
	Define random variables $S',Y',Z',Q'$ such that
	\begin{gather*}
		P(S'Y'Z'Q'=syzq) = P(SQ=sq) P(Y=y|SQ=sq) P(Z=z|SQ=sq).
	\end{gather*}
	Then $P(S'Y'=sy) = P(SY=sy)$ and $P(S'Z'=sz) = P(SZ=sz)$.  Since $SI_{*}$ only depends on the $(SY)$- and
	$(SZ)$-marginals, $SI_{*}(S;Y,Z) = SI_{*}(S';Y',Z')$.  Moreover,
	\begin{gather*}
		SI_{*}(S;Y,Z) = I(S;Q) = I(S';Q') 
		\le SI_{\GH}(S';Y',Z') \le SI_{*}(S';Y',Z'), 
	\end{gather*}
 {where the first inequality follows from \eqref{eq:def-GH} and the second one was discussed following \eqref{eq:def-SI*}.} 
The claim follows from this. 
 \end{itemize}
\end{proof}

\section{Conclusions}
\label{sec:conclusions}

We have studied measures that have been defined for bivariate information decompositions, asking whether they are continuous and/or additive. The only information decomposition that is both continuous and additive is~$I_{\BROJA}$. 

While there are many continuous information decompositions, it seems difficult to construct differentiable information {decompositions}: Currently, the only differentiable example is $I_{\IG}$ (which, however, is only defined in the interior of the probability simplex).  It would be interesting to know which other smoothness properties are satisfied by the proposed information decompositions, such as locking and asymptotic continuity. 

It also seems to be difficult to construct additive information decompositions, with $I_{\BROJA}$ and the G{\'a}cs-K{\"o}rner-based measures $I_{\cap}^{\wedge}$, $I_{\cap}^{\GH}$ and $I_{\cap}^{*}$ being the only ones. In contrast, many known information decompositions are additive under i.i.d.\ sequences. In the other direction, it would be worthwhile to have another look at stronger versions of additivity, such as chain rule-type properties. \citet{BROJ13:Shared_information} concluded that such chain rules prevent a straightforward extension of decompositions to the non-bivariate case along the lines of~\citet{WilliamsBeer:Nonneg_Decomposition_of_Multiinformation}. It has been argued (see, e.g., \citet{Rauh17:SSDeco}) that a general information decomposition likely needs a structure that differs from the proposal of ~\citep{WilliamsBeer:Nonneg_Decomposition_of_Multiinformation}, whence another look at chain rules may be worthwhile. Recent work \citep{Ay2019InformationDB} has proposed an additive decomposition based on a different lattice.

\subsection*{Acknowledgement} 
PB and GM have been supported by the ERC under the European Union’s Horizon 2020 research and innovation programme (grant agreement n\textsuperscript{o} 757983).

\bibliographystyle{abbrvnat}
\bibliography{info}

\end{document}